\documentclass{article}       

\usepackage{amsthm}
\usepackage{amsmath}
\usepackage{graphicx}
\usepackage{amssymb}
\usepackage{bold-extra}
\usepackage{amsthm}
\usepackage{lineno,hyperref}

\newtheorem{thm}{Theorem}[section]
\newtheorem{lma}[thm]{Lemma}
\newtheorem{cor}[thm]{Corollary}

\newtheorem{prop}[thm]{Proposition}
\newtheorem*{adef}{Definition}
\newtheorem{claim}{Claim}

\DeclareMathOperator{\tw}{tw}
\DeclareMathOperator{\conn}{conn}
\DeclareMathOperator{\adj}{adj}

\begin{document}

\newcommand{\leqfptT}{ $\leq^{\textrm{fpt}}_{\textrm{T}}$ }
\newcommand{\leqfptP}{ $\leq^{\textrm{fpt}}_{\textrm{pars}}$ }
\newcommand{\paramcount}[1]{$p$-\#\textsc{#1}}
\newcommand{\genprob}{Induced Subgraph With Property}

\title{The Parameterised Complexity of Counting Connected Subgraphs and Graph Motifs \thanks{Research supported by EPSRC grant ``Computational Counting''}}
\date{}
\author{Mark Jerrum\thanks{\texttt{m.jerrum@qmul.ac.uk}} and Kitty Meeks\thanks{\texttt{k.meeks@glasgow.ac.uk}} \thanks{Present address: School of Mathematics and Statistics, University of Glasgow}\\
\small{School of Mathematical Sciences, Queen Mary University of London}}
\maketitle

\begin{abstract}
We introduce a family of parameterised counting problems on graphs, \paramcount{\genprob}($\Phi$), which generalises a number of problems which have previously been studied.  This paper focusses on the case in which $\Phi$ defines a family of graphs whose edge-minimal elements all have bounded treewidth; this includes the special case in which $\Phi$ describes the property of being connected.  We show that exactly counting the number of connected induced $k$-vertex subgraphs in an $n$-vertex graph is \#W[1]-hard, but on the other hand there exists an FPTRAS for the problem; more generally, we show that there exists an FPTRAS for \paramcount{\genprob}($\Phi$) whenever $\Phi$ is monotone and all the minimal graphs satisfying $\Phi$ have bounded treewidth.  We then apply these results to a counting version of the \textsc{Graph Motif} problem.
\end{abstract}

\section{Introduction}

Parameterised counting problems were introduced by Flum and Grohe in \cite{flum04} and also independently by McCartin \cite{mccartin06}.  In this paper we focus on problems of the following form:
\\

\leftskip=1cm
\noindent
\textit{Input:} An $n$-vertex graph $G = (V,E)$, and $k \in \mathbb{N}$. \\
\textit{Parameter:} $k$. \\
\textit{Question:} How many (labelled) $k$-vertex subsets of $V$ induce graphs with a given property?
\\

\leftskip=0cm
It should be noted that, while the statement of this problem is concerned with \emph{induced} subgraphs, it also encompasses problems more often formulated in terms of counting subgraphs that are not necessarily induced.  For example, to count the number of $k$-vertex paths in $G$ (not necessarily induced), we would consider the labelled subgraph induced by $v_1,\ldots,v_k$ to have the desired property if and only if $v_iv_{i+1}$ is an edge for $1 \leq i \leq k-1$, regardless of what other edges may be present (and then divide the result by two, as this will count each path exactly twice).

Many problems of this form are known to be \#W[1]-hard (see Section \ref{complexity} for definitions of concepts from parameterised complexity), and thus are unlikely to be solvable exactly in time $f(k)n^{O(1)}$ for any function $f$.  A number of these \#W[1]-hard problems are in fact \emph{induced} subgraph counting problems: Chen and Flum \cite{chen07} demonstrated that problems of counting $k$-vertex induced paths and of counting $k$-vertex induced cycles are both \#W[1]-complete, and more generally Chen, Thurley and Weyer \cite{chen08} showed that it is \#W[1]-complete to count the number of induced subgraphs isomorphic to a given graph from the class $\mathcal{C}$ (\paramcount{Induced Subgraph Isomorphism}$(\mathcal{C})$) whenever $\mathcal{C}$ contains arbitrarily large graphs.  Most other subgraph counting problems previously studied in the literature can be described in the following way, for appropriate choices of a class of graphs $\mathcal{H}$:
\\

\hangindent=1cm
\paramcount{Sub}$(\mathcal{H})$ \\
\textit{Input:} A graph $G$ and an element $H \in \mathcal{H}$. \\
\textit{Parameter:} $k = |V(H)|$. \\
\textit{Question:} How many subgraphs (not necessarily induced) of $G$ are isomorphic to $H$? \\

Examples of \#W[1]-hard problems of this form include counting the number of $k$-vertex cliques (\paramcount{Clique} \cite{flum04}), paths (\paramcount{Path} \cite{flum04}), cycles (\paramcount{Cycle} \cite{flum04}) and matchings (\paramcount{Matching} \cite{radu13}).  Very recently, Curticapean and Marx \cite{radu14} proved a dichotomy result for \paramcount{Sub}$(\mathcal{H})$, demonstrating that the problem is \#W[1]-complete unless all the graphs in $\mathcal{H}$ have vertex-cover number bounded by some fixed constant, in which case the problem is fixed parameter tractable.

A natural question, therefore, is whether such counting problems, which are hard to solve exactly, can be efficiently approximated.  It is shown in \cite{arvind02} that there exists an efficient approximation scheme for \paramcount{Sub}($\mathcal{H}$) whenever $\mathcal{H}$ is a class of graphs having bounded treewidth,

In Section \ref{model} below, we introduce formally a family of parameterised counting problems which includes all the specific problems discussed above.  This family also includes the problem of counting the number of $k$-vertex connected induced subgraphs, a problem which we show to be \#W[1]-hard in Section \ref{exact}.  In Section \ref{approximate} we generalise the approximation result from \cite{arvind02}, showing that there exists an FPTRAS for the more general problem of counting the number of (labelled) $k$-vertex subsets of a graph satisfying a monotone property $\Phi$, provided that the edge-minimal graphs satisfying $\Phi$ all have bounded treewidth.  Examples of problems in this class for which there exists an FPTRAS include those of counting the number of $k$-vertex induced subgraphs that are connected, the number of $k$-vertex induced subgraphs that are Hamiltonian, and the number of $k$-vertex induced subgraphs that are not bipartite.  This last example contrasts with the result of Khot and Raman \cite{khot02} that deciding whether a graph contains an induced $k$-vertex subgraph that is bipartite is W[1]-hard.

Finally, in Section \ref{motif}, we apply some of these results to a counting version of the problem \textsc{Graph Motif}, introduced by Lacroix, Fernandes and Sagot \cite{lacroix06} in the context of metabolic networks.  The problem takes as input an $n$-vertex coloured graph, together with a \emph{motif} or multiset of colours $M$, and a solution is a subset $U$ of $|M|$ vertices such that the subgraph induced by $U$ is connected and the colour-(multi)set of $U$ is exactly $M$.  A counting version of this problem was studied by Guillemot and Sikora \cite{guillemot13}; we define and analyse a different natural counting version of \textsc{Graph Motif}, which is a more direct translation of the standard decision version into the counting world.

In the remainder of this section, we first introduce some notation in Section \ref{notation}, then introduce some key concepts in the study of parameterised counting complexity in Section \ref{complexity}, before giving formal definitions of the problems we consider in Section \ref{problems}.

\subsection{Notation}
\label{notation}

Given a graph $G = (V,E)$, and a subset $U \subset V$, we write $G[U]$ for the subgraph of $G$ induced by the vertices of $U$.  We denote by $\overline{G}$ the \emph{complement} of $G$, that is, $\overline{G} = (V,E')$ where $E' = V^{(2)} \setminus E$.  If $v \in V$, then $\Gamma(v)$ denotes the set of neighbours of $v$ in $G$.  For any $k \in \mathbb{N}$, we write $[k]$ as shorthand for $\{1,\ldots,k\}$, and denote by $S_k$ the set of all permutations on $[k]$, that is, injective functions from $[k]$ to $[k]$.  We write $V^{(k)}$ for the set of all subsets of $V$ of size exactly $k$, and $V^{\underline{k}}$ for the set of $k$-tuples $(v_1,\ldots,v_k) \in V^k$ such that $v_1,\ldots,v_k$ are all distinct. 

If $G$ is coloured by some colouring $\omega: V \rightarrow [k]$, we say that a subset $U \subset V$ is \emph{colourful} (under $\omega$) if, for every $i \in [k]$, there exists a unique vertex $u \in U$ such that $\omega(u) = i$; note that this can only be achieved if $U \in V^{(k)}$.  We write $\omega|_U$ for the restriction of $\omega$ to the set $U$; if $U$ is colourful under $\omega$ then $\omega|_U$ is a bijection.

Given graphs $G$ and $H$, a \emph{embedding} of $H$ in $G$ is an injective mapping $\theta: V(H) \rightarrow V(G)$ such that, for all $uv \in E(H)$, we have $\theta(u)\theta(v) \in E(G)$.

We will be considering labelled graphs, where a labelled graph is a pair $(H, \pi)$ such that $H$ is a graph and $\pi : [|V(H)|] \rightarrow V(H)$ is a bijection.  We write $\mathcal{L}(k)$ for the set of all labelled graphs on the vertex set $[k]$.  Given a graph $G = (V,E)$ and a $k$-tuple of vertices $(v_1,\ldots,v_k) \in V^{\underline{k}}$, $G[v_1,\ldots,v_k]$ denotes the labelled graph $(H,\pi)$ where $H = G[\{v_1,\ldots,v_k\}]$ and $\pi(i) = v_i$ for each $i \in [k]$.  We write $(H,\pi) \subseteq (H',\pi')$ if, for all $e=uv \in E(H)$, $\pi'(\pi^{-1}(u))\pi'(\pi^{-1}(v)) \in E(H')$.  Given a collection $\mathcal{H} \subseteq \mathcal{L}(k)$ of labelled graphs, we say that a graph $(H,\pi) \in \mathcal{H}$ is an \emph{edge-minimal} element of $\mathcal{H}$ if there is no $(H',\pi') \in \mathcal{H}$ such that $(H',\pi') \subseteq (H,\pi)$ and $H'$ has strictly fewer edges than $H$.  

We say that $(T,\mathcal{D})$ is a \emph{tree decomposition} of $G$ if $T$ is a tree and $\mathcal{D} = \{\mathcal{D}(t): t \in V(T)\}$ is a collection of non-empty subsets of $V(G)$ (or \emph{bags}), indexed by the nodes of $T$, satisfying:
\begin{enumerate}
\item $V(G) = \bigcup_{t \in V(T)} \mathcal{D}(t)$,
\item for every $e=uv \in E(G)$, there exists $t \in V(T)$ such that $u,v \in \mathcal{D}(t)$,
\item for every $v \in V(G)$, if $T(v)$ is defined to be the subgraph of $T$ induced by nodes $t$ with $v \in \mathcal{D}(t)$, then $T(v)$ is connected.
\end{enumerate}
The \emph{width} of the tree decomposition $(T,\mathcal{D})$ is defined to be $\max_{t \in V(T)} |\mathcal{D}(t)| - 1$, and the \emph{treewidth} of $G$, written $\tw (G)$, is the minimum width over all tree decompositions of $G$.

\subsection{Parameterised counting complexity}
\label{complexity}

In this section, we introduce key notions from parameterised counting complexity, which we will use in the rest of the paper.  A parameterised counting problem is a pair $(\Pi,\kappa)$ where, for some finite alphabet $\Sigma$, $\Pi: \Sigma^* \rightarrow \mathbb{N}_0$ is a function  and $\kappa: \Sigma^* \rightarrow \mathbb{N}$ is a parameterisation (a polynomial-time computable mapping).  An algorithm $A$ for a parameterised counting problem $(\Pi,\kappa)$ is said to be an \emph{fpt-algorithm} if there exists a computable function $f$ and a constant $c$ such that the running time of $A$ on input $I$ is bounded by $f(\kappa(I))|I|^c$.  Problems admitting an fpt-algorithm are said to belong to the class FPT.

To understand the complexity of parameterised counting problems, Flum and Grohe \cite{flum04} introduce two kinds of reductions between such problems.

\begin{adef}
Let $(\Pi,\kappa)$ and $(\Pi',\kappa')$ be parameterised counting problems.
\begin{enumerate}
\item An fpt parsimonious reduction from $(\Pi,\kappa)$ to $(\Pi',\kappa')$ is an algorithm that computes, for every instance $I$ of $\Pi$, an instance $I'$ of $\Pi'$ in time $f(\kappa(I))\cdot |I|^c$ such that $\kappa'(I') \leq g(\kappa(I))$ and 
$$\Pi(I) = \Pi'(I')$$ 
(for computable functions $f,g: \mathbb{N} \rightarrow \mathbb{N}$ and a constant $c \in \mathbb{N}$).  In this case we write $(\Pi,\kappa)$ \emph{\leqfptP} $(\Pi',\kappa')$.

\item An fpt Turing reduction from $(\Pi,\kappa)$ to $(\Pi',\kappa')$ is an algorithm $A$ with an oracle to $\Pi'$ such that
\begin{enumerate}
\item $A$ computes $\Pi$,
\item $A$ is an fpt-algorithm with respect to $\kappa$, and
\item there is a computable function $g:\mathbb{N} \rightarrow \mathbb{N}$ such that for all oracle queries ``$\Pi'(I') = ?$'' posed by $A$ on input $x$ we have $\kappa'(I') \leq g(\kappa(I))$.
\end{enumerate}
In this case we write $(\Pi,\kappa)$ \emph{\leqfptT} $(\Pi',\kappa')$.
\end{enumerate}

\end{adef}

Using these notions, Flum and Grohe introduce a hierarchy of parameterised counting complexity classes, \#W[$t$], for $t \geq 1$; this is the analogue of the W-hierarchy for parameterised decision problems.  In order to define this hierarchy, we need some more notions related to satisfiability problems.  

The definition of levels of the hierarchy uses the following problem, where $\psi$ is a first-order formula with a free relation variable of arity $s$.
\\

\hangindent=1cm
\paramcount{WD}$_{\psi}$ \\
\textit{Input:} A relational structure\footnote{The relational structure $\mathcal{A}$ of vocabulary $\tau$ of relations consists of the \emph{universe} $A$ together with an interpretation $R^{\mathcal{A}}$ of every relation $R$ in $\tau$.} $\mathcal{A}$ and $k \in \mathbb{N}$. \\
\textit{Parameter:} $k$. \\
\textit{Question:} How many relations $S \subseteq A^s$ of cardinality $|S|=k$ are such that $\mathcal{A} \models \psi(S)$ (where $A$ is the universe of $\mathcal{A}$)? \\

If $\Psi$ is a class of first-order formulas, then \paramcount{WD}-$\Psi$ is the class of all problems \paramcount{WD}$_{\psi}$ where $\psi \in \Psi$.  The classes of first-order formulas $\Sigma_t$ and $\Pi_t$, for $t \geq 0$, are defined inductively.  Both $\Sigma_0$ and $\Pi_0$ denote the class of quantifier-free formulas, while, for $t \geq 1$, $\Sigma_t$ is the class of formulas
$$\exists x_1 \ldots \exists x_k \psi,$$
where $\psi \in \Pi_{t-1}$, and $\Pi_t$ is the class of formulas
$$\forall x_1 \ldots \forall x_k \psi,$$
where $\psi \in \Sigma_{t-1}$.  We are now ready to define the classes \#W[$t$], for $t \geq 1$.

\begin{adef}[\cite{flum04,flumgrohe}]
For $t \geq 1$, $\#W[t]$ is the class of all parameterised counting problems that are fpt parsimonious reducible to \paramcount{WD}-$\Pi_t$.
\end{adef} 

Unless FPT=W[1], there does not exist an algorithm running in time $f(k)n^{O(1)}$ for any problem that is hard for the class \#W[1] under either fpt parsimonious reductions or fpt Turing reductions.  In the setting of this paper, a parameterised counting problem will be considered to be intractable if it is \#W[1]-hard with respect to either form of reduction.

In \cite{flumgrohe}, Flum and Grohe also define a counting version of the A-hierarchy for parameterised problems; this turns out to be easier to use for some of our purposes.  The definition is in terms of the following model-checking problem, where $C$ is a class of structures and $\Psi$ a class of formulas.
\\

\hangindent=1cm
\paramcount{MC}$(C,\Psi)$ \\
\textit{Input:} A structure $\mathcal{A} \in C$ and a formula $\psi \in \Psi$. \\
\textit{Parameter:} $|\psi|$. \\
\textit{Question:} What is $|\psi(\mathcal{A})|$? \\

Here, $\psi(\mathcal{A})$ is the set of tuples $(a_1,\ldots,a_k) \in A^k$ such that $\psi(a_1,\ldots,a_k)$ is true in $\mathcal{A}$, where $k$ is the number of free variables in $\psi$, and $A$ the universe of $\mathcal{A}$.  If $C$ is the class of all structures, we write simply \paramcount{MC}($\Psi$).  The counting analogue of the A-hierarchy is then defined as follows.

\begin{adef}[\cite{flumgrohe}]
For all $t \geq 1$, $\#A[t]$ is the class of all parameterised counting problems reducible to \paramcount{MC}$(\Pi_{t-1})$ by an fpt parsimonious reduction.
\end{adef}

It is known that the first levels of these two hierarchies for parameterised counting problems coincide:
\begin{thm}[\cite{flumgrohe}]
\#W[1] = \#A[1].
\label{W=A}
\end{thm}
Thus, to prove that a problem belongs to \#W[1] (=\#A[1]) it suffices to show that it is reducible, under fpt parsimonious reductions, to \paramcount{MC}($\Pi_0$).

When considering approximation algorithms for parameterised counting problems, an ``efficient'' approximation scheme is an FPTRAS, as introduced by Arvind and Raman \cite{arvind02}; this is the analogue of a FPRAS (fully polynomial randomised approximation scheme) in the parameterised setting.
\begin{adef}
An FPTRAS for a parameterised counting problem $\Pi$ with parameter $k$ is a randomised approximation scheme that takes an instance $I$ of $\Pi$ (with $|I| = n$), and real numbers $\epsilon > 0$ and $0 < \delta < 1$, and in time $f(k) \cdot g(n,1/\epsilon,\log(1/\delta))$ (where $f$ is any computable function, and $g$ is a polynomial in $n$, $1/\epsilon$ and $\log(1 / \delta)$) outputs a rational number $z$ such that
$$\mathbb{P}[(1-\epsilon)\Pi(I) \leq z \leq (1 + \epsilon)\Pi(I)] \geq 1 - \delta.$$
\end{adef}

\subsection{Problems considered}
\label{problems}

In this section we begin by introducing a general family of parameterised counting problems on graphs, in which the goal is to count $k$-tuples of vertices that induce subgraphs with particular properties.  We then give formal definitions of the problems we will consider in Sections \ref{exact} and \ref{approximate}.  Subject to appropriate rescaling, our model can be regarded as a generalisation of many problems that involve counting labelled subgraphs, including \paramcount{Cycle} \cite{flum04}, \paramcount{Path} \cite{flum04}, \paramcount{StrEmb}($\mathcal{C}$) \cite{chen08} and \#$k$-\textsc{Matching} \cite{radu13}.  Induced subgraph problems that are invariant under relabelling of vertices have also been studied in the literature on parameterised counting, including \paramcount{Clique} \cite{flum04}, and can be regarded as instances of a sub-family of problems in our model.  A more detailed discussion of how problems previously studied in the literature, including \paramcount{Sub}$(\mathcal{H})$, can be expressed in the language of this model is given in \cite{bddlayers}.

\subsubsection{The model}
\label{model}

Let $\Phi$ be a family $(\phi_1,\phi_2,\ldots)$ of functions $\phi_k: \mathcal{L}(k) \rightarrow \{0,1\}$, such that the function mapping $k \mapsto \phi_k$ is computable.  For any $k$, we write $\mathcal{H}_{\phi_k}$ for the set $\{(H,\pi) \in \mathcal{L}(k): \phi_k(H,\pi) = 1\}$, and set $\mathcal{H}_{\Phi} = \bigcup_{k \in \mathbb{N}} \mathcal{H}_{\phi_k}$.

We then define the following problem.
\\

\hangindent=1cm
\paramcount{\genprob}($\Phi$) \\
\textit{Input:} A graph $G = (V,E)$ and $k \in \mathbb{N}$.\\
\textit{Parameter:} $k$. \\
\textit{Question:} What is the cardinality of the set $\{(v_1,\ldots,v_k) \in V^{\underline{k}}:$ \\
$\phi_k(G[v_1,\ldots,v_k]) = 1 \}$? \\

Observe that we can equivalently regard this problem as that of counting induced labelled $k$-vertex subgraphs that belong to $\mathcal{H}_{\Phi}$.  Note that this generalises the problem \paramcount{Sub}$(\mathcal{H})$, as the latter problem only permits counting copies (not necessarily induced) of one particular graph $H$, whereas \paramcount{\genprob}$(\Phi)$ allows us to count all $k$-vertex subgraphs having some more complicated property.

We say that $\Phi$ is a \emph{monotone} property if, for every $k$, whenever $\phi_k(H,\pi) = 1$ for some $(H,\pi) \in \mathcal{L}(k)$, and $(H',\pi') \in \mathcal{L}(k)$ with $(H,\pi) \subseteq (H',\pi')$, then we also have $\phi_k(H',\pi') = 1$.  We further describe $\Phi$ as a \emph{symmetric} property if the value of $\phi_k(H,\pi)$ depends only on the graph $H$ and not on the labelling of the vertices; this corresponds to ``unlabelled'' graph problems, such as \paramcount{clique}.  We can define a related problem for symmetric properties:
\\

\hangindent=1cm
\paramcount{Induced Unlabelled Subgraph With Property}($\Phi$) \\
\textit{Input:} A graph $G = (V,E)$ and $k \in \mathbb{N}$.\\
\textit{Parameter:} $k$. \\
\textit{Question:} What is the cardinality of the set $\{\{v_1,\ldots,v_k\} \in V^{(k)}: \phi_k(G[v_1,\ldots,v_k]) = 1 \}$? \\

For any symmetric property $\Phi$, the output of \paramcount{\genprob}$(\Phi)$ is exactly $k!$ times the output of \paramcount{Induced Unlabelled Subgraph With Property}$(\Phi)$.  The unlabelled version is less general than the labelled version, as this only allows us count induced subgraphs having some particular property rather than, for example, all (not necessarily induced) copies of some fixed graph $H$.  As an example, the labelled version can express problems such as \paramcount{Matching}, whereas the former would only allow us to count $k$-vertex induced subgraphs that contain a perfect matching (ignoring the fact that any one $k$-vertex induced subgraph may contain many perfect matchings).  

Note that, for any $\phi_k \in \Phi$, the problem of determining the cardinality of the set $\{(v_1,\ldots,v_k) \in V^{\underline{k}}: \phi_k(G[v_1,\ldots,v_k]) = 1 \}$ can easily be expressed as an instance of \paramcount{MC}($\Pi_0$); thus, by Theorem \ref{W=A}, we obtain the following result:
\begin{prop}
For any $\Phi$, the problem \paramcount{\genprob}($\Phi$) belongs to \#W[1].  If $\Phi$ is symmetric, then the same is true for \paramcount{Induced Unlabelled Subgraph With Property}$(\Phi)$.
\label{in-W}
\end{prop}
In order to give an fpt parsimonious reduction from \paramcount{Induced Unlabelled Subgraph With Property}$(\Phi)$ to \paramcount{MC}$(\Pi_0)$, we can imitate a technique used in \cite[Lemma 14.31]{flumgrohe}, introducing an additional relation in our structure which imposes an order on the elements in order to ensure that each unlabelled subset is counted exactly once.

\subsubsection{Problem definitions}
\label{probdefs}

In Section \ref{exact}, we consider the following problem.
\\

\hangindent=1cm
\paramcount{Connected Induced Subgraph} \\
\textit{Input:} A graph $G = (V,E)$ and $k \in \mathbb{N}$.\\
\textit{Parameter:} $k$. \\
\textit{Question:} For how many subsets $U \in V^{(k)}$ is $G[U]$ connected? \\

This problem is clearly symmetric, and can be regarded as a particular case of the general problem \paramcount{Induced Unlabelled Subgraph With Property}($\Phi$) introduced above.  Let $\mathcal{T}_k$ be the set of all trees on $k$ vertices with vertices labelled $1,\ldots,k$, and then set $\Phi^{\conn} = (\phi_1^{\conn},\phi_2^{\conn},\ldots)$, with
\begin{equation}
\phi_k^{\conn}(H,\pi) = \bigvee_{T \in \mathcal{T}_k} \bigwedge_{\{j,l\} \in E(T)} \adj_H(\pi(j),\pi(l)),
\label{phi_conn}
\end{equation}
where $\adj_H(u,w) = 1$ if and only if $uw \in E(H)$, and $\adj_H(u,w) = 0$ otherwise.  The tuples $(v_1,\ldots,v_k) \in V^{\underline{k}}$ such that $\phi_k(G[v_1,\ldots,v_k]) = 1$ are then exactly the tuples such that $G[\{v_1,\ldots,v_k\}]$ is connected.

In Section \ref{approximate}, we consider the more general problem of solving \paramcount{\genprob}($\Phi$), whenever $\Phi = (\phi_1,\phi_2,\ldots)$ is a monotone property and there exists a positive integer $t$ such that, for each $\phi_k$, all edge-minimal labelled $k$-vertex graphs $(H,\pi)$ such that $\phi_k(H,\pi) = 1$ satisfy $\tw(H) \leq t$.  Note that \paramcount{Connected Induced Subgraph} is a special case of this more general problem: the set of edge-minimal labelled $k$-vertex graphs $(H,\pi)$ such that $\phi_k^{\conn}(H,\pi) = 1$ is in fact precisely the set of labelled trees on $k$ vertices.

\section{\paramcount{Connected Induced Subgraph} is \#W[1]-complete}
\label{exact}

In this section, we prove the following result.

\begin{thm}
\paramcount{Connected Induced Subgraph} is \#W[1]-complete under fpt Turing reductions.
\label{exact-hard}
\end{thm}

We begin in Section \ref{lattices} by noting some background results we will need for the proof, before demonstrating \#W[1]-hardness with a series of fpt Turing reductions in Section \ref{reduction}.

\subsection{Lattices and M\"{o}bius functions}
\label{lattices}

In Section \ref{reduction} we will need to consider the lattice formed by partitions of a $k$-element set, with a partial order given by the refinement relation.  A partition of a set $X$ is a set of disjoint subsets $X_1,\ldots,X_r$ of $X$ such that $X = X_1 \cup \cdots \cup X_r$; the subsets $X_1,\ldots,X_r$ are called the \emph{blocks} of the partition.  A partition $P'$ \emph{refines} the partition $P$ if every block of $P'$ is contained in some block of $P$ (note that this ordering is the opposite of that more commonly used for partitions).  In this section we recall some existing results about lattices on arbitrary posets and also more specifically about partition lattices, which we will use in the proof of Lemma \ref{can-invert}.

A lattice is a partially ordered set $(\mathcal{P},\leq)$ satisfying the condition that, for any two elements $x,y \in \mathcal{P}$, both the \emph{meet} and \emph{join} of $x$ and $y$ also belong to $\mathcal{P}$, where the meet of $x$ and $y$, written $x \wedge y$, is defined to be the unique element $z$ such that 
\begin{enumerate}
\item $z \leq x$ and $z \leq y$, and
\item for any $w$ such that $w \leq x$ and $w \leq y$, we have $w \leq z$,
\end{enumerate}
and the join of $x$ and $y$, $x \vee y$, is correspondingly defined to be the unique element $z'$ such that 
\begin{enumerate}
\item $x \leq z'$ and $y \leq z'$, and
\item for any $w$ such that $x \leq w$ and $y \leq w$, we have $z' \leq w$.
\end{enumerate}
We denote by $\hat{1}$ and $\hat{0}$ respectively the ``top'' and ``bottom'' elements of the lattice (the ``top'' element is the unique $x \in \mathcal{P}$ such that, for all $y \in \mathcal{P}$, $y \leq x$, and the bottom element is defined symmetrically).

We will make use of the M\"{o}bius function $\mu$ on a poset, which is defined inductively by
\[
\mu(x,y) = \begin{cases}
				1	& \text{if } x = y \\
				- \sum_{z:x \leq z < y} \mu(x,z) & \text{for } x < y \\
				0	& \text{otherwise}.
			\end{cases}
\]
In Lemma \ref{can-invert}, we will consider a so-called \emph{meet-matrix} on a partition lattice, and make use of the following lemma (which follows immediately from a result of Haukkanen \cite[Cor.~2]{haukkanen96}).  
\begin{lma}
Let $x_1,\ldots, x_n$ be the elements of a finite lattice $(\mathcal{P}, \leq)$, let $f: \mathcal{P} \rightarrow \mathbb{C}$ be a function, and let $A = (a_{ij})_{1 \leq i,j \leq n}$ be the matrix given by $a_{ij} = f(x_i \wedge x_j)$.  Then
$$\det(A) = \prod_{i=1}^n \sum_{x_k \leq x_i} f(x_k)\mu(x_k,x_i),$$
where $\mu$ is the M\"{o}bius function for $\mathcal{P}$.
\label{determinant}
\end{lma}

To make use of this result in Section \ref{reduction}, we will need to be able to calculate certain values of the M\"{o}bius function for the special case in which $\mathcal{P}$ is a partition lattice, ordered so that $x \leq y$ whenever $y$ refines $x$; in fact, due to our choice of $f$ below, it will suffice to be able to compute $\mu(\hat{0},x)$ for each partition $x$.  To do this, we will use the following lemma, which is an immediate consequence of two results of Rota \cite[Section 3, Prop.~3 and Section 7, Prop.~3]{rota64}.

\begin{lma}
Let $\mathcal{P}$ be the lattice of partitions of a set with $n$ elements, where $x \leq y$ if and only if $y$ refines $x$.  If $x \in \mathcal{P}$ is of rank $r$, then 
$$\mu(\hat{0},x) = (-1)^{r}r!$$
where the rank of $x$ is equal to the number of blocks of $x$ minus one.
\label{mu-value}
\end{lma}
Since the rank of any element lies in the range $[0,n-1]$, we obtain the following immediate corollary.
\begin{cor}
Let $\mathcal{P}$ be the lattice of partitions of a set with $n$ elements, where $x \leq y$ if and only if $y$ refines $x$.  Then, for all $x \in \mathcal{P}$, $\mu(\hat{0},x) \neq 0$.
\label{mu-non-0}
\end{cor}

\subsection{The reduction}
\label{reduction}

In this section, we prove Theorem \ref{exact-hard}.  The main work in this proof is to give an fpt Turing reduction from \paramcount{Multicolour Independent Set} to \paramcount{Multicolour Connected Induced Subgraph}, where these two problems are defined as follows.
\\

\hangindent=1cm
\paramcount{Multicolour Independent Set} \\
\textit{Input:} A $k$-coloured graph $G = (V,E)$ and an integer $k$.\\
\textit{Parameter:} $k$. \\
\textit{Question:} For how many colourful subsets $U \in V^{(k)}$ is $U$ an independent set in $G$? \\

\hangindent=1cm
\paramcount{Multicolour Connected Induced Subgraph} \\
\textit{Input:} A $k$-coloured graph $G = (V,E)$ and an integer $k$.\\
\textit{Parameter:} $k$. \\
\textit{Question:} For how many colourful subsets $U \in V^{(k)}$ is $G[U]$ connected? \\

In this reduction we will set up a system of equations, argue that, with an oracle to \paramcount{Multicolour Connected Induced Subgraph}, we can compute the entries, and show that the system can be solved to give the number of colourful independent sets in our graph.  Throughout, we will need to switch between considering colourful subsets of vertices and partitions of $[k]$.  Let $\mathcal{P}_k$ be the set of all partitions of the set $[k]$; thus the cardinality of $\mathcal{P}_k$ is precisely the $k^{th}$ \emph{Bell number}, $B_k$.  We consider these partitions to be partially ordered by the refinement relation, so $P_i \leq P_j$ if $P_j$ refines $P_i$.  Set $P_1 = \hat{0} = \{[k]\}$ and $P_{B_k} = \hat{1} = \{\{1\},\ldots,\{k\}\}$.
	 
Suppose that $G = (V,E)$ is the $k$-coloured graph in an instance of \paramcount{Multicolour Independent Set}.    Given a multicolour subset $U \in V^{(k)}$, we set $P(U)$ to be the partition of $[k]$ in which $i,j \in [k]$ belong to the same set of the partition if and only if the vertices of $U$ with colours $i$ and $j$ belong to the same connected component in $G[U]$.  

We define a function $f: \mathcal{P}_k \rightarrow \{0,1\}$ such that, for any partition $P \in \mathcal{P}_k$,
\[
f(P) = \begin{cases}
		   1  & \text{if $P = \{[k]\}$} \\
		   0  & \text{otherwise.}
		 \end{cases}
\]

In the following lemma we set up our system of equations, and use results from Section \ref{lattices} to demonstrate that the system can be solved to determine the number of colourful independent sets in our graph.

\begin{lma}
Given all values of $\sum_{U \in V^{(k)}} f(P(U) \wedge P')$ for $P' \in \mathcal{P}_k$, we can compute the number of colourful independent sets in $G$ in time $h(k)$, where $h$ is some computable function.
\label{can-invert}
\end{lma}
\begin{proof}
For $1 \leq i \leq B_k$, let $N_i$ be the number of subsets $U \in V^{(k)}$ such that $P(U) = P_i$.  Since $P_{B_k} = \{\{1\},\ldots,\{k\}\}$, our goal is then to calculate $N_{B_k}$. 

Let $A = (a_{ij})_{0 \leq i,j \leq B_k}$ be the matrix given by $a_{ij} = f(P_i \wedge P_j)$.  We first claim that $A \cdot \mathbf{N} = \mathbf{z}$ where
$$\mathbf{N} =  (N_0, \ldots, N_{B_k})^T,$$
and
$$\mathbf{z} =  (z_1, \ldots, z_{B_k})^T$$
with $z_i = \sum_{U \in V^{(k)}} f(P(U) \wedge P_i)$.

To see that this is true, observe that the $i^{th}$ element of $A \cdot \mathbf{N}$ is 
\begin{align*}
\sum_{j=1}^{B_k} a_{ij} N_j & = \sum_{j=1}^{B_k} f(P_i \wedge P_j) N_j \\
							& = \sum_{j=1}^{B_k} \sum_{\substack{U \in V^{(k)} \\ P(U) = P_j}} f(P_i \wedge P_j) \\
							& = \sum_{U \in V^{(k)}} f(P_i \wedge P(U)) \\
							& = z_i,
\end{align*}
as required.  Thus it suffices to prove that the matrix $A$ is nonsingular, as then we can compute $N_{B_k}$ as the last element of $A^{-1} \cdot \mathbf{z}$.  

To see that this is indeed the case, first note that, by Lemma \ref{determinant}, 
\begin{align*}
\det(A) & = \prod_{j=1}^{B_k} \sum_{P_i \leq P_j} f(P_i) \mu(P_i,P_j) \\
        & = \prod_{j=1}^{B_k} \mu(\hat{0},P_j).
\end{align*}
Thus it suffices to verify that all values of $\mu(\hat{0},P_j)$ for $P_j \in \mathcal{P}$ are non-zero; but this follows immediately from Corollary \ref{mu-non-0}.  Hence $\det(A) \neq 0$ and $A$ is nonsingular, as required.
\end{proof}

Now we show that, with an oracle to \paramcount{Multicolour Connected Induced Subgraph}, we can compute the values required to set up the equations in the previous lemma, completing the reduction from \paramcount{Multicolour Independent Set} to \paramcount{Multicolour Connected Induced Subgraph}.

\begin{lma}
There exists a computable function $g$ such that, with an oracle to \paramcount{Multicolour Connected Induced Subset}, the value of $\sum_{U \in V^{(k)}} f(P(U) \wedge P_i)$ can be computed, for every $P_i \in \mathcal{P}_k$, in time $g(k) \cdot n^{O(1)}$.  Moreover, for every oracle call, the parameter value is at most $2k$.
\label{can-compute}
\end{lma}
\begin{proof}
We begin by considering how to compute the values of $\sum_{U \in V^{(k)}} f(P(U) \wedge P_i)$ for a single $P_i \in \mathcal{P}_k$.  Suppose $P_i = \{X_1, \ldots, X_{\ell}\}$, where each $X_j \subset [k]$.

We construct a new coloured graph $G_i$, with vertex set $V(G_i) = V(G) \cup \{x_1,\ldots,x_{\ell}\}$, and where the colouring $c$ of $V(G)$ is extended to $V(G_i)$ by setting $c(x_j) = k+j$ for $1 \leq j \leq \ell$.  $G_i$ has edge-set 
$$E(G_i) = E(G) \cup \bigcup_{1 \leq j \leq \ell} \{x_jv: v \text{ has colour $d$ for some $d \in X_j$}\}.$$

Suppose that $W \subset V(G_i)$ is a multicoloured subset of $V(G_i)$ (so $|W| = k + \ell$, and all vertices in $W$ have distinct colours), and set $U = W \cap V(G)$.  Note that, in order for $W$ to be colourful, we must have $W \setminus U = V(G_i) \setminus V(G) = \{x_1,\ldots,x_{\ell}\}$.  We make the following claim.

\begin{claim}
$G_i[W]$ is connected if and only if $f(P(U) \wedge P_i) = 1$, that is, if and only if the finest partition that is refined by both $P(U)$ and $P_i$ is in fact $\{[k]\}$.
\label{connected-equiv}
\end{claim}
\begin{proof}[Proof of claim]
Suppose first that $G_i[W]$ is connected.  It suffices to prove that, for any $u_1,u_2 \in U$, $c(u_1)$ and $c(u_2)$ belong to the same block of $P(U) \wedge P_i$.  Note that, by connectedness of $G_i[W]$, there must exist a path in $G_i[W]$ from $u_1$ to $u_2$.  We now proceed by induction on the length of a shortest $u_1$-$u_2$ path in $G_i[W]$.  

For the base case, suppose that there is at most one internal vertex on such a path.  In this case, either $u_1$ and $u_2$ belong to the same connected component of $G[U]$ (in which case we are done, since by definition $c(u_1)$ and $c(u_2)$ then belong to the same block of $P(U)$ and hence $P(U) \wedge P_i$), or else there is a single internal vertex $x_j \in W \setminus U$ lying on this path.  Thus $u_1,u_2 \in \Gamma(x_j)$, implying by the construction of $G_i$ that $c(u_1),c(u_2) \in X_j$, so $c(u_1)$ and $c(u_2)$ belong to the same block of $P_i$ and hence of $P(U) \wedge P_i$.  This completes the proof of the base case.

We may now assume that there are at least two internal vertices on a shortest $u_1$-$u_2$ path, and that the result holds for any $u_1',u_2'$ that are connected by a shorter path in $G_i[W]$.  Since there are at least two internal vertices on the shortest $u_1$-$u_2$ path in $G_i[W]$, and no two vertices in $W \setminus U$ are adjacent, there must be some vertex $u_3 \in U \setminus \{u_1,u_2\}$ that lies on this path.  But then there exists a shorter $u_1$-$u_3$ path in $G_i[W]$, implying by the inductive hypothesis that $c(u_1)$ and $c(u_3)$ belong to the same block of $P(U) \wedge P_i$.  Similarly, we see that $c(u_2)$ and $c(u_3)$ belong to the same block of $P(U) \wedge P_i$, and hence it must be that $c(u_1)$ and $c(u_2)$ belong to the same block of $P(U) \wedge P_i$, as required.

We now consider the reverse implication.  Suppose that $G_i[W]$ is not connected, so this graph has connected components with vertex sets $W_1, \ldots, W_r$, where $r \geq 2$.  We claim that the partition $P_W = \{c(W_1 \cap U),\ldots,c(W_r \cap U)\}$ of $[k]$ is refined by both $P(U)$ and $P_i$; hence 
$$P(U) \wedge P_i \geq P_W > \hat{0},$$
so $P(U) \wedge P_i \neq \hat{0}$, as required.  To see that this claim holds, it suffices to check that, for every block $X$ of $P(U)$ or $P_i$, the vertices having colours from $X$ belong to the same component of $G_i[W]$.  If $X \in P(U)$ then this follows immediately, since by definition blocks of $P(U)$ are sets of colours that appear in the same connected component of $G[U]$, and so must certainly belong to the same connected component of $G_i[W]$.  Suppose therefore that $X \in P_i$.  But then $X = X_j$ for some $1 \leq j \leq \ell$, and so all vertices of $U$ with colours from $X$ are adjacent to $x_j$, and hence belong to the same component of $G_i[W]$.  This completes the proof of the claim.  

It therefore follows that $G_i[W]$ is connected if and only if $P(U) \wedge P_i = \hat{0}$, as required.
\renewcommand{\qedsymbol}{$\square$ (Claim \ref{connected-equiv})}
\end{proof}

Thus, by Claim \ref{connected-equiv}, $\sum_{U \in V^{(k)}} f(P(U) \wedge P_i)$ is exactly equal to the number of colourful connected subsets in $G_i$.  Thus, with $B_k < k^k$ calls to an oracle to \paramcount{Multicolour Connected Induced Subgraph}, we can compute the value of $\sum_{U \in V^{(k)}} f(P(U) \wedge P_i)$ for every $P_i \in \mathcal{P}_k$; for each call, the parameter value $k+\ell$ is at most $2k$.
\end{proof}

Using Lemmas \ref{can-compute} and \ref{can-invert}, it is now straightforward to prove the main result of this section.

\begin{proof}[Proof of Theorem \ref{exact-hard}]
The fact that \paramcount{Connected Induced Subgraph} $\in$ \#W[1] follows immediately from Proposition \ref{in-W}, so it suffices to prove that the problem is \#W[1]-hard.  To do this, we give a sequence of fpt Turing reductions from \paramcount{Clique}, shown to be \#W[1]-complete in \cite{flum04}.  
\paragraph*{\paramcount{Clique} \leqfptT \paramcount{Multicolour Independent Set}} 
For this reduction, we mimic the proof of Fellows et al.~\cite{fellows09} that \textsc{Multicolour Clique} is W[1]-complete.  Let $G=(V,E)$ be the graph in an instance of \paramcount{Clique}, with parameter $k$.  Now define $G'$ to be the cartesian product $\overline{G} \times K_k$, in which each vertex in $\overline{G}$ (the complement of $G$) is ``blown up'' to a $k$-clique; the vertices of each such $k$-clique are given distinct colours $\{1,\ldots,k\}$.  It is straightforward to check that if $\alpha$ is the number of multicolour independent sets in $G'$, then the number of $k$-cliques in $G$ is exactly equal to $\alpha / k!$.

\paragraph*{\paramcount{Multicolour Independent Set} \\ \leqfptT \paramcount{Multicolour Connected Induced Subgraph}} 
The reduction follows immediately from Lemmas \ref{can-invert} and \ref{can-compute}.

\paragraph*{\paramcount{Multicolour Connected Induced Subgraph} \\ \leqfptT \paramcount{Connected Induced Subgraph}} 
The number of multicoloured connected induced subgraphs in a graph $G$ can be computed by inclusion-exclusion from the numbers of connected induced subgraphs in the $2^k$ subgraphs of $G$ induced by different combinations of colour-classes.  (Inclusion-exclusion methods have been used in a similar way in, for example, \cite{chen08,dalmau04}.)  Suppose the graph $G$ is coloured with colours $[k]$, and for any $C \subseteq [k]$ let $G_C$ be the subgraph of $G$ induced by the vertices with colours belonging to $C$.  Then, if $N_k(H)$ denotes the number of connected induced $k$-vertex subgraphs in $H$, the number of colourful connected induced subgraphs in $G$ is exactly
$$\sum_{\emptyset \neq C \subseteq [k]} (-1)^{k - |C|}N_k(G_C).$$

Combining these reductions, we have \paramcount{Clique} \leqfptT \paramcount{Connected Induced Subgraph}, and so \paramcount{Connected Induced Subgraph} is \#W[1]-hard under fpt Turing reductions, as required.
\end{proof}

\section{Approximating \paramcount{\genprob}($\Phi$)}
\label{approximate}

In contrast to the hardness result of the previous section, we now give a positive result about the approximability of a class of parameterised counting problems that includes \paramcount{Connected Induced Subgraph}, as well as, for example, the problems of counting the number of induced $k$-vertex Hamiltonian subgraphs, and that of counting the number of induced $k$-vertex non-bipartite subgraphs.  This is in fact a special case of a more general result (Theorem \ref{motif-general-approx}) which will be given in Section \ref{motif} below; we prove Theorem \ref{approx-count} first as it introduces all the techniques but with slightly less complexity than is required for Theorem \ref{motif-general-approx}.

\begin{thm}
Let $\Phi = (\phi_1,\phi_2,\ldots)$ be a monotone property, and suppose there exists a positive integer $t$ such that, for each $\phi_k$, all edge-minimal labelled $k$-vertex graphs $(H,\pi)$ such that $\phi_k(H) = 1$ satisfy $\tw(H) \leq t$.  Then there is an FPTRAS for \paramcount{\genprob}($\Phi$).
\label{approx-count}
\end{thm}

Recall from Section \ref{probdefs} that \paramcount{Connected Induced Subgraph} is a special case of this problem, so we obtain the following immediate corollary to Theorem \ref{approx-count}.

\begin{cor}
There is an FPTRAS for \paramcount{Connected Induced Subgraph}.
\end{cor}

The proof of Theorem \ref{approx-count} is adapted from the proof of Arvind and Raman \cite{arvind02} that there is an FPTRAS for \paramcount{Sub($\mathcal{H}$)} whenever $\mathcal{H}$ is a class of graphs of bounded treewidth.  We begin in Section \ref{approx-background} by summarising the existing results we will use, and then in Section \ref{FPTRAS} give a proof of the existence of an FPTRAS in the setting of Theorem \ref{approx-count}.

\subsection{Background}
\label{approx-background}

The algorithm we describe in the next section uses random sampling to count approximately, and relies heavily on the parameterised version of the Karp-Luby result \cite{karp83} on this subject, given by Arvind and Raman.

\begin{thm}[{\cite[Thm.~1]{arvind02}}]
For every positive integer $n$, and for every integer $0 \leq k \leq n$, let $U_{n,k}$ be a finite universe, whose elements are binary strings of length $n^{O(1)}$.  Let $\mathcal{A}_{n,k} = \{A_1, \ldots, A_m\}  \subseteq U_{n,k}$ be a collection of $m = m_{n,k}$ given sets, with $m_{n,k} = l(k)n^{O(1)}$ for some function $l$, let $g:\mathbb{N} \rightarrow \mathbb{N}$ be a computable function and let $d > 0$ be a constant with the following conditions:
\begin{enumerate}
\item There is an algorithm that computes $|A_i|$ in time $g(k)n^d$, for each $i$, and every $\mathcal{A}_{n,k}$.
\item There is an algorithm that samples uniformly at random from $A_i$ in time $g(k)n^d$, for each $i$, and every $\mathcal{A}_{n,k}$.
\item There is an algorithm that takes $x \in U_{n,k}$ as input and determines whether $x \in A_i$ in time $g(k)n^d$, for each $i$, and every $\mathcal{A}_{n,k}$.
\end{enumerate}
Then there is an FPTRAS for estimating the size of $A = A_1 \cup \cdots \cup A_m$.  In particular, for $\epsilon = 1/g(k)$, and $\delta = 1/2^{n^{O(1)}}$, the running time of the FPTRAS algorithm is $(g(k))^{O(1)}n^{O(1)}$.
\label{K-L}
\end{thm}

In proving that there exists an FPTRAS for the problem \paramcount{Sub}($\mathcal{H}$) when $\mathcal{H}$ is a class of graphs of bounded treewidth, Arvind and Raman prove two further results which we will use in Section \ref{FPTRAS}.  Firstly, they give an algorithm to compute the number of colourful copies of a $k$-vertex graph $H$ (of bounded treewidth) in a $k$-coloured graph $G$; it should be noted here that the copies of $H$ are not necessarily induced.

\begin{lma}[{\cite[Lemma 1]{arvind02}}]
Let $G = (V,E)$ be a graph on $n$ vertices that is $k$-coloured by some colouring $f:V(G) \rightarrow [k]$, and let $H$ be a $k$-vertex graph of treewidth $t$ that is $k$-coloured by some colouring $\pi$ such that $H$ is colourful.  Then there is an algorithm taking time $O(c^{t^3}k + n^{t+2}2^{t^2/2})$ to exactly compute the cardinality of the set $\{K: K$ is a colourful $k$-vertex subgraph of $G$ and $K$ is colour-preserving isomorphic to $H$ coloured by $\pi \}$, where $c>0$ is some constant.
\label{arvind1}
\end{lma}

The graph $H_1$ with colouring $\omega_1$ is said to be \emph{colour-preserving isomorphic} to $H_2$ with colouring $\omega_2$ if there exists an isomorphism $\theta$ from $H_1$ to $H_2$ such that, for all $u \in V(H)$,  $\omega_1(u) = \omega_2(\theta(u))$.  We will more generally say that a function $\theta$ from the vertices of the graph $H_1$, coloured by $\omega_1$, to the vertices of the graph $H_2$, coloured by $\omega_2$, is \emph{colour-preserving} if, for all $u \in V(H_1)$, $\omega_1(u) = \omega_2(\theta(u))$.

Secondly, they describe an algorithm to sample uniformly at random from the set of colourful copies of $H$ in $G$.

\begin{lma}[{\cite[Lemma 2]{arvind02}}]
Let $G=(V,E)$ be a graph on $n$ vertices that is $k$-coloured by some colouring $f: V(G) \rightarrow [k]$, and let $H$ be a $k$-vertex graph of treewidth $t$ that is $k$-coloured by some colouring $\pi$ such that $H$ is colourful.  Then there is an algorithm taking time $O(c^{t^3}k + n^{t + O(1)}2^{t^2/2})$ time to sample uniformly at random from the set $\{K : K$ is a colourful $k$-vertex subgraph of $G$ under the colouring $f$ and $K$ is colour-preserving isomorphic to $H$ coloured by $\pi \}$.
\label{arvind2}
\end{lma}

The approximation algorithm in \cite{arvind02} also uses the concept of \emph{$k$-perfect} families of hash functions.  A family $\mathcal{F}$ of hash functions from $[n]$ to $[k]$ is said to be $k$-perfect if, for every subset $A \subset [n]$ of size $k$, there exists $f \in \mathcal{F}$ such that the restriction of $f$ to $A$ is injective.  In the following section, we will use the following bound on the size of such a family of hash functions, proved in \cite{alon95}.

\begin{thm}
For all $n, k \in \mathbb{N}$ there is a $k$-perfect family $\mathcal{F}_{n,k}$ of hash functions from $[n]$ to $[k]$ of cardinality $2^{O(k)} \cdot \log n$.  Furthermore, given $n$ and $k$, a representation of the family $\mathcal{F}_{n,k}$ can be computed in time $2^{O(k)} \cdot  n \log n$.
\label{hash-functions}
\end{thm}

\subsection{An FPTRAS for \paramcount{\genprob}($\Phi$)}
\label{FPTRAS}

In this section we use Theorem \ref{K-L} to give a proof of Theorem \ref{approx-count}, that is, we show that there exists an FPTRAS for \paramcount{\genprob}($\Phi$) whenever $\Phi = (\phi_1,\phi_2,\ldots)$ is a monotone property and there exists a positive integer $t$ such that, for each $\phi_k$, all edge-minimal labelled $k$-vertex graphs $(H,\pi)$ such that $\phi_k(H) = 1$ satisfy $\tw(H) \leq t$.

When considering this problem, we will take $U_{n,k}$ to be the set of all $k$-tuples of $[n]$; thus an element of $U_{n,k}$ can be regarded as a choice of a $k$-tuple of vertices in an $n$-vertex graph.  Our goal is to approximate the cardinality of the set $A$, where
$$A = \{(v_1,\ldots,v_k) \in V^{\underline{k}}: \phi_k(G[v_1,\ldots,v_k]) = 1 \}.$$
Thus, in order to make use of Theorem \ref{K-L} to prove the existence of an FPTRAS for \paramcount{\genprob}($\Phi$), we need to express $A$ as a union of sets $A_1, \ldots, A_{m_{n,k}}$ which satisfy the conditions of the theorem.

First, we will write $A$ as a union of sets indexed by a family of $k$-perfect hash functions from $V$ to $[k]$, which we shall regard as vertex-colourings of $G$.  Recall from Theorem \ref{hash-functions} that we can fix such a family $\mathcal{F}$ with $|\mathcal{F}| = 2^{O(k)}\log n$.  Since, by definition of a family of $k$-perfect hash functions, there must exist for every $U \in V^{(k)}$ some element $f_U \in \mathcal{F}$ such that the restriction of $f_U$ to $U$ is injective, it is clear that we can write
$$A = \bigcup_{f \in \mathcal{F}} A_f,$$
where we set
\begin{align*}
A_f = \{(v_1,\ldots,v_k) \in V^{\underline{k}}: \quad & \phi_k(G[v_1,\ldots,v_k]) = 1 \text{ and } \{f(v_1),\ldots,f(v_k)\} = [k] \}. \\
\end{align*}
We can further write $A_f$ as a (disjoint) union of smaller sets, conditioning on the precise injective colouring of $\{v_1,\ldots,v_k\}$ under $f$ (recall that $S_k$ denotes the set of permutations on $[k]$):
$$A_f = \bigcup_{\sigma \in S_k} A_{f,\sigma},$$
where
\begin{align*}
A_{f,\sigma} = \{(v_1,\ldots,v_k) \in V^{\underline{k}}: \quad & \phi_k(G[v_1,\ldots,v_k]) = 1, \text{ and, for each } \\
												 & 1 \leq i \leq k, f(v_i) = \sigma(i)\}.
\end{align*}
In order to obtain an FPTRAS, we will need to use the assumption of Theorem \ref{FPTRAS}, namely that $\Phi$ is monotone and that every edge-minimal labelled subgraph $(H,\pi) \in \mathcal{L}(k)$ such that $\phi_k(H,\pi) = 1$ satisfies $\tw(H) \leq t$.  If we write $\mathcal{H}_k$ for the set of edge-minimal labelled subgraphs satisfying $\phi_k$, this characterisation of $\phi_k$ implies that $\phi_k(G[v_1,\ldots,v_k]) = 1$ if and only if there is some $(H,\pi) \in \mathcal{H}_k$ such that $(H,\pi) \subseteq G[v_1,\ldots,v_k]$.  We can therefore write $A_{f,\sigma}$ as a union over sets indexed by elements of $\mathcal{H}_k$:
$$A_{f,\sigma} = \bigcup_{(H,\pi) \in \mathcal{H}_k} A_{f,\sigma,(H,\pi)},$$
where
\begin{align*}
A_{f,\sigma,(H,\pi)} = \{(v_1,\ldots,v_k) \in V^{\underline{k}}: \quad & (H,\pi) \subseteq G[v_1,\ldots,v_k] \text{ and,}\\
														 & \text{for each } 1 \leq i \leq k, f(v_i) = \sigma(i)\}.
\end{align*}
In words, the pair of conditions in the definition above can be restated as follows: the mapping taking the vertex $\pi(i)$ of $H$ (for each $1 \leq i \leq k$) to the vertex in $\{v_1,\ldots,v_k\}$ which receives colour $\sigma(i)$ under $f$ is, in fact, an embedding.  If we then equip $H$ with a colouring $\omega$, where $\omega = \sigma \circ \pi^{-1}$, we can equivalently describe this embedding as the mapping which takes each vertex of $u$ of $H$ to the unique vertex $v_i$ such that $\omega(u) = f(v_i)$, so the mapping is the unique colour-preserving bijection from $V(H)$ to $\{v_1,\ldots,v_k\}$ (with respect to colourings $\omega$ and $f$).  With this characterisation, it is clear that the condition that $(H,\pi) \subseteq G[v_1,\ldots,v_k]$ is exactly the same as the requirement that $H$ with colouring $\omega$ is colour-preserving isomorphic to some subgraph $K$ of $G[\{v_1,\ldots,v_k\}]$.  Thus,
\begin{align*}
A_{f,\sigma,(H,\pi)} = \{(v_1,\ldots,v_k) \in V^{\underline{k}}: \quad & \exists K \subseteq G[\{v_1,\ldots,v_k\}] \text{ such that $H$ with } \\
														 & \text{colouring $\omega = \sigma \circ \pi^{-1}$ is colour-preserving}\\
														 & \text{isomorphic to $K$ with colouring $f$}\}.
\end{align*}
These are the sets that will make up the collection $\mathcal{A}_{n,k}$ in Theorem \ref{K-L}; more precisely, we set 
$$\mathcal{A}_{n,k} = \{A_{f,\sigma,(H,\pi)}: f \in \mathcal{F}, \sigma \in S_k, (H,\pi) \in \mathcal{H}_k\},$$
and it then follows from the reasoning above that
\begin{equation}
A = \bigcup \mathcal{A}_{n,k}.
\label{A-union}
\end{equation}
Note that
$$|\mathcal{A}_{n,k}| \leq 2^{O(k)}\log n \cdot k! \cdot 2^{\binom{k}{2}} = l(k)n^{O(1)}$$
for an appropriate function $l$ (since we can choose $\mathcal{F}$ with $|\mathcal{F}| = 2^{O(k)}\log n$, there are $k!$ permutations on a set of size $k$, and there are $2^{\binom{k}{2}}$ labelled graphs on a fixed set of $k$ vertices), as required in the premise of Theorem \ref{K-L}.

Before going on to demonstrate that this collection of sets $\mathcal{A}_{n,k}$ satisfies the three conditions of Theorem \ref{K-L}, it will be useful to make a further observations about the elements of $\mathcal{A}_{n,k}$.  Note that there can be at most one subgraph $K \subseteq G[\{v_1,\ldots,v_k\}]$ such that $H$ with colouring $\omega$ is colour-preserving isomorphic to $K$ (since the colourings determine precisely the mapping between the two graphs), so if we set 
\begin{align*}
A'_{f,\sigma,(H,\pi)} = \{K: \quad & K \text{ is a colourful $k$-vertex subgraph of $G$ under} \\
								   & \text{the colouring $f$, and $K$ is colour-preserving} \\
								   & \text{isomorphic to $H$ with colouring $\omega = \sigma \circ \pi^{-1}$} \},
\end{align*}
we have
\begin{equation}
|A_{f,\sigma,(H,\pi)}| = |A'_{f,\sigma,(H,\pi)}|.
\label{sets-equal}
\end{equation}
Moreover, we can define a bijection $\theta: A'_{f,\sigma,(H,\pi)} \rightarrow A_{f,\sigma,(H,\pi)}$ by setting
\begin{equation}
\theta(K) = \left(((f|_{V(K)})^{-1} \circ \sigma)(1),\ldots,((f|_{V(K)})^{-1} \circ \sigma)(k)\right).
\label{bijection-defn}
\end{equation}

We are now ready to show that the sets $\mathcal{A}_{n,k}$ defined above do satisfy the three conditions of Theorem \ref{K-L}.  The first two conditions will follow easily from results proved in \cite{arvind02}.

\begin{lma}
For each $\mathcal{A}_{n,k}$ and every $A_i \in \mathcal{A}_{n,k}$, there exists  an algorithm that computes $|A_i|$ in time $g_1(k)n^{d_1}$, where $d_1$ is an integer and $g_1: \mathbb{N} \rightarrow \mathbb{N}$ is a computable function, for each $i$ and every $\mathcal{A}_{n,k}$.
\label{condition1}
\end{lma}
\begin{proof}
Recall from \eqref{sets-equal} that, for each $A_{f,\sigma,(H,\pi)} \in \mathcal{A}_{n,k}$, we have 
$$|A_{f,\sigma,(H,\pi)}| = |A'_{f,\sigma,(H,\pi)}|,$$
and so it suffices to compute $|A'_{f,\sigma,(H,\pi)}|$ in the permitted time.  Since, by assumption, $H$ has treewidth at most $t$, we can immediately apply Lemma \ref{arvind1} to see that there exists an algorithm to compute the cardinality of $A'_{f,\sigma,(H,\pi)}$ in time at most $O(c^{t^3}k + n^{t+2}2^{t^2/2})$ where $c > 0$ is a constant.
\end{proof}

We now show that the second condition is satisfied.

\begin{lma}
For each $\mathcal{A}_{n,k}$ and every $A_i \in \mathcal{A}_{n,k}$, there exists  an algorithm that samples uniformly at random from $A_i$ in time $g_2(k)n^{d_2}$, where $d_2$ is an integer and $g_2: \mathbb{N} \rightarrow \mathbb{N}$ is a computable function, for each $i$ and every $\mathcal{A}_{n,k}$.
\label{condition2}
\end{lma}
\begin{proof}
It follows immediately from the definition of $A'_{f,\sigma,(H,\pi)}$, together with the assumption that $H$ has treewidth at most $t$, that, by Lemma \ref{arvind2}, there is an algorithm taking time $O(c^{t^3}k + n^{t+O(1)}2^{t^2/2})$ (where $c > 0$ is a constant) to sample uniformly at random from $A'_{f,\sigma,(H,\pi)}$.  Since $\theta$ (as defined in \eqref{bijection-defn}) gives a bijection from $A'_{f,\sigma,(H,\pi)}$ to $A_{f,\sigma,(H,\pi)}$, applying $\theta$ to the output of this sampling algorithm will give an element of $A_{f,\sigma,(H,\pi)}$ chosen uniformly at random; note that applying $\theta$ will require additional time depending only on $k$.
\end{proof}

Thus, in order to apply Theorem \ref{K-L}, it remains to check that our sets satisfy the third condition; we demonstrate this in the following lemma.

\begin{lma}
For each $\mathcal{A}_{n,k}$ and every $A_i \in \mathcal{A}_{n,k}$, there is an algorithm that takes $\mathbf{v} \in U_{n,k}$ as input and determines whether $\mathbf{v} \in A_i$ in time $g_3(k)n^{d_3}$, where $d_3$ is an integer and $g_3: \mathbb{N} \rightarrow \mathbb{N}$ is a computable function.
\label{condition3}
\end{lma}
\begin{proof}
For any $\mathbf{v}=(v_1,\ldots,v_k) \in U_{n,k}$, in order to determine whether $\mathbf{v} \in A_i$ for any given $A_i = A_{f,\sigma,(H,\pi)}$, it suffices to check whether both the following conditions are satisfied:
\begin{enumerate}
\item for each $1 \leq i \leq k$, $f(v_i) = \sigma(i)$, and
\item $G[v_1,\ldots,v_k] \supseteq (H,\pi)$.
\end{enumerate}
The first of these two conditions can clearly be verified in time depending only on $k$.  For the second condition we need to check, for every edge $e=uv \in E(H)$, whether $v_{\pi^{-1}(u)}x_{\pi^{-1}(v)} \in E(G)$; this can also be done in time depending only on $k$.  The result follows immediately.
\end{proof}
 
With these three lemmas, we can prove Theorem \ref{approx-count}.

\begin{proof}[Proof of Theorem \ref{approx-count}]
We wish to approximate the cardinality of a set $A$ which, by \eqref{A-union}, can be written as a union of sets $A_1,\ldots,A_m$ (where $m_{n,k}$ is $l(k)n^{O(1)}$ for some function $l$).  It follows from Lemmas \ref{condition1}, \ref{condition2} and \ref{condition3} that, if we set $g(k) = \max \{g_1(k),g_2(k),g_3(k)\}$ and $d = \max \{d_1,d_2,d_3\}$, these sets satisfy the three conditions of Theorem \ref{K-L}; it therefore follows immediately that there exists an FPTRAS for estimating the size of $A$, in other words there exists an FPTRAS for \paramcount{\genprob}($\Phi$).
\end{proof}

\section{Application to \paramcount{Graph Motif}}
\label{motif}

The \textsc{Graph Motif} problem was first introduced by Lacroix, Fernandes and Sagot \cite{lacroix06} in the context of metabolic network analysis, and is defined as follows.
\\

\hangindent=1cm
\textsc{Graph Motif} \\
\textit{Input:} A vertex-coloured graph $G$ and a multiset of colours $M$. \\
\textit{Question:} Does $G$ have a connected subset of vertices whose multiset of colours equals $M$? \\

This decision problem, and a number of variations, have since been studied extensively (\cite{betzler08, dondi07, dondi09, fellows-motif11, guillemot13}).  The problem is known to be NP-complete in general \cite{lacroix06}, and remains NP-complete even if the input is restricted so that $G$ is a tree of maximum degree three and $M$ is a set rather than a multiset \cite{fellows-motif11}.  However, the decision problem is fixed parameter tractable when parameterised by the motif size $|M|$ \cite{fellows-motif11}.

It is natural to consider counting versions of the \textsc{Graph Motif} problem, and counting the number of occurrences of a given motif in a graph has applications in determining whether a motif is over- or under-represented in a biological network with respect to the null hypothesis \cite{lacroix09}.  In \cite{guillemot13}, Guillemot and Sikora consider the following parameterised counting version of the problem.
\\

\hangindent=1cm
\paramcount{XMGM} \\
\textit{Input:} A graph $G=(V,E)$, a colouring $c$ of $V$, and a multiset of colours $M$. \\
\textit{Parameter:} $k = |M|$. \\
\textit{Question:} How many $k$-vertex trees in $G$ have a multiset of colours equal to $M$?\\

The authors prove that this problem is \#W[1]-hard in the case that $M$ is a multiset, but is fixed parameter tractable when $M$ is in fact a set (\#XCGM).

In \paramcount{XMGM}, the output is the number of connected induced subgraphs of $G$ having colour-set exactly equal to $M$, where each such subgraph is weighted by its number of spanning trees.  In this section we consider a more direct translation of \textsc{Graph Motif} into the counting world, in which the goal is to compute simply the total number of connected induced subgraphs having the desired colour-set.
\\

\hangindent=1cm
\paramcount{Graph Motif} \\
\textit{Input:} A graph $G = (V,E)$, a colouring $c$ of $V$, and a multiset of colours $M$.\\
\textit{Parameter:} $k = |M|$. \\
\textit{Question:} How many subsets $U \subset V^{(k)}$ are such that $G[U]$ is connected and the multiset of colours assigned to $U$ is exactly $M$?\\

We adapt results from Sections \ref{exact} and \ref{approximate} to show that
\begin{itemize}
\item \paramcount{Graph Motif} is \#W[1]-hard, even in the case that $M$ is a set, and
\item there exists an FPTRAS for \paramcount{Graph Motif}.
\end{itemize}

Our hardness result is obtained by means of a trivial reduction from \paramcount{Multicolour Connected Induced Subgraph}, shown to be \#W[1]-complete in Section \ref{exact-hard}.

\begin{thm}
\paramcount{Graph Motif} is \#W[1]-hard, even when $M$ is a set.
\end{thm}

Now we show that it is possible to approximate \paramcount{Graph Motif}, for any input $(G,M)$.  In fact, we prove that there exists an FPTRAS for the following generalisation of \paramcount{\genprob}$(\Phi)$ when $\Phi$ satisfies the conditions of Theorem \ref{approx-count}.
\\

\hangindent=1cm
\paramcount{Induced Coloured Subgraph With Property}$(\Phi)$ (\paramcount{ICSWP}$(\Phi)$) \\
\textit{Input:} A graph $G = (V,E)$, a colouring $c$ of $V$, and a multiset of colours $M$.\\
\textit{Parameter:} $k = |M|$. \\
\textit{Question:} What is the cardinality of the set $\{(v_1,\ldots,v_k) \in V^{\underline{k}}:$ \\
$\phi_k(G[v_1,\ldots,v_k]) = 1$ and $\{c(v_1),\ldots,c(v_k)\} = M\}$?\\

\begin{thm}
Let $\Phi = (\phi_1,\phi_2,\ldots)$ be a monotone property, and suppose there exists a positive integer $t$ such that, for each $\phi_k$, all edge-minimal labelled $k$-vertex graphs $(H,\pi)$ such that $\phi_k(H) = 1$ satisfy $\tw(H) \leq t$.  Then there exists an FPTRAS for \paramcount{ICSWP}$(\Phi)$.
\label{motif-general-approx}
\end{thm}
\begin{proof}
Once again, we use Theorem \ref{K-L} to demonstrate the existence of an FPTRAS for this problem.   As before, we will take $U_{n,k}$ to be the set of $k$-element subsets of $n$.  To make use of Theorem \ref{K-L},  we need to express 
\begin{align*}
B = \{(v_1,\ldots,v_k) \in V^{\underline{k}}: \quad & \phi_k(G[v_1,\ldots,v_k]) = 1 \text{ and} \\
									  & \{c(v_1),\ldots,c(v_k)\} = M \}
\end{align*}
as the union of some collection of sets $\mathcal{B}_{n,k}$ that satisfy the conditions of the theorem.  Applying the same reasoning as in the proof of Theorem \ref{approx-count} (and using the same notation), we see that 
$$B = \bigcup \{B_{f,\sigma,(H,\pi)}: f \in \mathcal{F}, \sigma \in S_k, (H,\pi) \in \mathcal{H}_k \},$$
where
\begin{align*}
B_{f,\sigma,(H,\pi)} = \{(v_1,\ldots,v_k) \in V^{\underline{k}}: \quad & (H,\pi) \subseteq G[v_1,\ldots,v_k],\{c(v_1),\ldots,c(v_k)\} = M,\\
														 &  \text{ and, for each } 1 \leq i \leq k, f(v_i) = \sigma(i)\}.
\end{align*}
Now, if we set $\mathcal{D}$ to be the set of all bijective mappings $d:[k] \rightarrow M$, we can write 
$$B_{f,\sigma,(H,\pi)} = \bigcup_{d \in \mathcal{D}} B_{f,\sigma,(H,\pi),d},$$
where
\begin{align*}
B_{f,\sigma,(H,\pi),d} = \{ (v_1,\ldots,v_k) \in V^{\underline{k}} : \quad & (H,\pi) \subseteq G[v_1,\ldots,v_k] \text{ and, for each } 1 \leq i \leq k, \\
				& f(v_i) = \sigma(i) \text{ and } c(v_i) = (d \circ f)(v_i) \}.
\end{align*}
So if we set 
$$\mathcal{B}_{n,k} = \{B_{f,\sigma,(H,\pi),d} : f \in \mathcal{F}, \sigma \in S_k, (H,\pi) \in \mathcal{H}_{\phi_k}, d \in \mathcal{D}\},$$
we have
$$B = \bigcup \mathcal{B}_{n,k}.$$
Note that $|\mathcal{B}_{n,k}| \leq k^k \cdot |\mathcal{A}_{n,k}|$ (with $\mathcal{A}_{n,k}$ as in the proof of Theorem \ref{approx-count}), and so we clearly have that $m_{n,k} = \tilde{l}(k)n^{O(1)}$ for some function $\tilde{l}$.  It remains to check that the sets $B_j \in \mathcal{B}_{n,k}$ satisfy the conditions of Theorem \ref{K-L}.

To see that the first two conditions hold, we make use of Lemmas \ref{condition1} and \ref{condition2}.  Observe that the set $B_{f,\sigma,(H,\pi),d}$, calculated with respect to the graph $G$ and its colouring $c$, is precisely equal to the set $A_{f,\sigma,(H,\pi)}$ (as defined in the proof of Theorem \ref{approx-count}) if instead of considering the graph $G$ we consider the graph 
$$G_{f,d} = G[\{v \in V: c(v) = (d \circ f)(v) \}].$$
Note that, given $f$ and $d$, the graph $G_{f,d}$ can be computed from $G$ in time $O(n^2)$, so it follows from Lemmas \ref{condition1} and \ref{condition2} that
\begin{itemize}
\item for each $\mathcal{B}_{n,k}$ and every $B_j \in \mathcal{B}_{n,k}$, there is an algorithm that computes $|B_j|$ in time $g_1(k)n^{d_1 + 1}$, and
\item for each $\mathcal{B}_{n,k}$ and every $B_j \in \mathcal{B}_{n,k}$, there is an algorithm that samples uniformly at random from $B_j$ in time $g_2(k)n^{d_2 + 1}$.
\end{itemize}

Thus it remains only to check that the third condition holds.  For any $k$-tuple $(v_1,\ldots,v_k) \in V^{\underline{k}}$ and any set $B_{f,\sigma,(H,\pi),d} \in \mathcal{B}_{n,k}$, we know by Lemma \ref{condition3} that we can check in time $g_3(k)n^{d_3}$ whether $(v_1,\ldots,v_k)$ belongs to the related set $B_{f,\sigma,(H,\pi)}$.  But $(v_1,\ldots,v_k) \in B_{f,\sigma,(H,\pi),d}$ if and only if we have both that $(v_1,\ldots,v_k) \in B_{f,\sigma,(H,\pi)}$ and that, for each $1 \leq i \leq k$, $c(v_i) = (d \circ f)(v_i)$; the time required to check this second condition clearly depends only on $k$.  Hence there exists a function $g_4: \mathbb{N} \rightarrow \mathbb{N}$ so that, for each $\mathcal{B}_{n,k}$ and every $B_j \in \mathcal{B}_{n,k}$, there is an algorithm that takes $\mathbf{v} \in U_{n,k}$ as input and determines whether $\mathbf{v} \in B_j$ in time $g_4(k)n^{d_3}$.

Hence, setting $g = \max \{g_1(k),g_2(k),g_4(k)\}$ and $d = \max \{d_1,d_2,d_3\}$, all the conditions of Theorem \ref{K-L} are satisfied, and therefore there exists an FPTRAS for \paramcount{ICSWP}$(\Phi)$.
\end{proof}

This result easily implies the existence of an FPTRAS for \paramcount{Graph Motif}.

\begin{cor}
There exists an FPTRAS for \paramcount{Graph Motif}.
\end{cor}
\begin{proof}
Setting $\Phi = \Phi^{\conn}$ (as in \eqref{phi_conn}), it is clear both that $\Phi$ satisfies the conditions of Theorem \ref{motif-general-approx} and also that the output of \paramcount{ICSWP}$(\Phi)$ will be exactly $k!$ times the number of connected induced $k$-vertex subgraphs whose vertices have multiset of colours equal to $M$ (with the overcounting due to the number of distinct possible labellings of a $k$-vertex subgraph).  Thus we know from Theorem \ref{motif-general-approx} that there exists an FPTRAS for \paramcount{ICSWP}$(\Phi)$ in this situation, and so obtain an FPTRAS for \paramcount{Graph Motif} simply by dividing the output of the first algorithm by $k!$.
\end{proof}

\section{Conclusions and Open Problems}

We have shown that the problem \paramcount{Connected Induced Subgraph} is \#W[1]-hard, but that on the other hand there exists an FPTRAS for a more general problem \paramcount{\genprob}($\Phi$), where $\Phi$ is a monotone property such that the edge-minimal graphs satisfying $\Phi$ all have bounded treewidth.  We then adapted these results to show that a natural counting version of the problem \textsc{Graph Motif} is \#W[1]-hard, but has an FPTRAS.

We finish with two natural related open questions.
\begin{enumerate}
\item Are all (non-trivial) special cases of the class of problems covered by Theorem \ref{approx-count} \#W[1]-hard?
\item Does there exist any problem \paramcount{\genprob}($\Phi$), where $\Phi$ is a monotone property but the edge-minimal graphs that satisfy $\Phi$ do not all have bounded treewidth, such that \paramcount{\genprob}$(\Phi)$ admits an FPTRAS?
\end{enumerate}

\bibliography{../param_counting_refs} 

\end{document}